\documentclass[showkeys,useAMS,twocolumn,nofootinbib,pra]{revtex4-1}
\usepackage[latin9]{inputenc}
\usepackage{graphics,graphicx}
\usepackage[breaklinks=true]{hyperref}
\usepackage{breakcites}
\usepackage{mathtools}
\usepackage{amsmath}
\usepackage{amsthm}
\usepackage{amssymb}
\usepackage{xcolor}
\usepackage{mathrsfs}

\makeatletter
\theoremstyle{plain}
\newtheorem*{thm}{\protect\theoremname}

\usepackage{dsfont}
\newcommand{\id}{\mathds{1}}
\newcommand{\Tr}{\text{Tr}}
\newcommand{\tr}{\text{Tr}}

\newcommand{\eq}[1]{\begin{equation}
    \begin{aligned}#1\end{aligned}
\end{equation}}
\newcommand{\ket}[1]{\left|#1\right\rangle}
\newcommand{\bra}[1]{\left\langle#1\right|}
\newcommand{\rk}{\text{rank}}
\newcommand{\ot}{\otimes}
\makeatother

\usepackage[english]{babel}

\providecommand{\theoremname}{Theorem}

\makeatletter
\newcommand{\printfnsymbol}[1]{%
  \textsuperscript{\@fnsymbol{#1}}%
}
\makeatother

\begin{document}

\title{Operational symmetries of entangled states}

\author{Ilan Tzitrin\normalfont\textsuperscript{1}}
\email{itzitrin@physics.utoronto.ca}
\author{Aaron Z. Goldberg\normalfont\textsuperscript{1}}
\email{goldberg@physics.utoronto.ca}
\author{Jesse C. Cresswell\normalfont\textsuperscript{1}}
\email{jcresswe@physics.utoronto.ca}

\affiliation{Department of Physics, University of Toronto, Toronto, ON, M5S 1A7}

\begin{abstract}
    Quantum entanglement obscures the notion of local operations; there exist quantum states for which all local actions on one subsystem can be equivalently realized by actions on another. We characterize the states for which this fundamental property of entanglement {does and does not hold}, including multipartite and mixed states. Our results lead to a method for quantifying entanglement based on operational symmetries and has connections to quantum steering, envariance, the Reeh-Schlieder theorem, {and classical entanglement.}
\end{abstract}
\maketitle

Entanglement -- touted by Schr\"odinger as the defining feature of quantum theory \cite{Schrodinger1935} -- is often phrased as the inability to describe certain states of multiple particles by local specifications of each one. But the significance of entanglement and Einstein, Podolsky, and Rosen's (EPR's) unease with it \cite{EPR1935} is perhaps elucidated in its more colloquial articulation, phrased positively: the ability of a local action \emph{here} to immediately reflect in a local action \emph{there}, no matter the distance between the two. 
Investigations of entanglement have since grown from its role as a mere curiosity \cite{Gisin1991,Mermin1993,Zurek2003,ArndtHornberger2014} to a resource \cite{Horodeckietal2009,Modietal2012,Howardetal2014,GoldbergJames2018Euler},  {expanding our understanding of its fundamental physical and metaphysical implications. }

Following EPR's paradox of local measurements causing instantaneous collapses of distant states \cite{EPR1935}, Schr\"odinger quickly realized that a local experimenter can use certain entangled states to ``steer'' the measurement results of a distant local experimenter \cite{Schrodinger1935}. An entangled state can be steered if a local experimenter cannot explain their measurement results even by assuming an unknown local state \cite{Wisemanetal2007}. These steerable states form a subset of those that cannot be described by any additional local variables \cite{Uolaetal2019}, where the latter violate Bell inequalities \cite{Bell1964}. All of these constructions need only a single set of measurements to show their deviation from classical predictions. In this work we investigate systems for which \textit{any} local operation can be detected by a separate experimenter. Our goal is to derive insight into the remarkable ramifications of quantum entanglement by formalizing the replicability of local actions among the subsystems of a composite state.

 Symmetries between local operations play an important role in applications such as environment-induced invariance, or \emph{envariance} \cite{Zurek2003}, where they can be used to derive Born's Rule from the no-collapse assumption of quantum mechanics \cite{Zurek2005,Herbut2007}. States are envariant under \emph{specific} local unitaries on a given system if they can be undone by local unitaries on the environment; here, we find states that are envariant under \emph{all} local operations. These states are intriguing because any local operation on them can be remotely negated, so that any observer must be completely ignorant of the local state.

Reproducibility of local operations on different subsystems is also a generic characteristic of quantum field theories (QFT) that is less well-known in the quantum information community. Working in the vacuum state on Minkowski spacetime, the Reeh-Schlieder theorem \cite{Reeh1961} entails that any operator supported in a local region can be reproduced by another operator in a different local region, potentially at spacelike separations. This result implies that the vacuum is entangled between any set of local regions, and has implications for the localizability of quantum fields \cite{Knight1961,Haag1996}. Our motivation is also to understand whether or not states in quantum mechanical systems obey an analogue of this cornerstone result.

\color{black}
Our investigation begins with the following question: What is the set of states $\left |\psi \right\rangle$ on a bipartite Hilbert space $\mathcal{H}_A\ot\mathcal{H}_B$ for which a linear operator $U\ot \id$ acting only on subsystem $A$ can be reproduced by an operation $\id\ot V$ acting only on subsystem $B$? An example of such a state is the Bell state of two qubits, $\ket{\Phi_0}\equiv\frac{1}{\sqrt{2}}\ket{00}+\ket{11}$, for which
$
U \otimes \id \ket{\Phi_0} = \id \otimes U^T\ket{\Phi_0}$ for any $U$.
We call this type of relationship an \emph{operational symmetry}. 


Consider the case where $A$ and $B$ are each represented by a $d$-dimensional qudit. We can express any pure 2-qudit state $\left|\psi\right\rangle$ in the Schmidt basis as
\begin{equation}\label{eq:compState}
    \left|\psi\right\rangle=\sum_{i=1}^{d}\sigma_{i}\left|i\right\rangle_{A}\left|i\right\rangle_{B},\ \ \sigma_i\geq\sigma_{i+1}\geq 0,\ \ \sum_i\sigma_i^2=1\, .
\end{equation}
The actions of general operators $U^A=U\otimes \id$ on $A$ and $V^B=\id\otimes V$ on  $B$ can be written in this basis as
\begin{align}\label{eq:UAaction}
    \begin{aligned}
    U^A\left|\psi\right\rangle
    &=\sum_{i,k=1}^{d}\sigma_{i}U_{ki}\left|k\right\rangle_{A}\left|i\right\rangle_{B},\\
    V^B\left|\psi\right\rangle&=\sum_{i,k=1}^{d}\sigma_{i}V_{ki}\left|i\right\rangle_{A}\left|k\right\rangle_{B}.
    \end{aligned}
\end{align}
These actions are equivalent if and only if 
\eq{U\Sigma=\Sigma V^T,
\label{eq:matrixconditions2}} where $\Sigma$ is the diagonal matrix of Schmidt coefficients.
Hence, the condition on $\left|\psi\right\rangle$ for a related operation $V_B$ to exist for a given $U_A$ is simple: the Schmidt matrix, $\Sigma$, must be invertible. 
In turn, this is equivalent to all of the Schmidt coefficients being nonzero, or the Schmidt rank being maximal, in which case we say the state is \emph{fully entangled} \cite{Witten2018}.

Restricting $U$ to be unitary, we can determine the class of states $\ket{\psi}$ that have related operations $V$ that are also unitary. It is necessary and sufficient for $\ket{\psi}$ to be maximally entangled; i.e., all of the Schmidt coefficients $\sigma_i$ must be equal \cite{Bennettetal1996,Gisin1998}. For maximally entangled $\ket{\psi}$, $\Sigma\propto\id$, so $V=U^T$ shares the unitarity of $U$. Unitarity for $V$, on the other hand, requires $V^\dagger V=\id$ for all $V=\Sigma U^T\Sigma^{-1}$. This implies $\left[U,\Sigma^2\right]=0$ for all unitary matrices $U$, which then implies $\Sigma\propto\id$, by Schur's first lemma \cite{Ramond2010}. Given some state, the unitarity of $V$ for an entire set of generators of an irreducible representation of the unitary group, such as the generalized Pauli matrices, is a necessary and sufficient condition for the state to be maximally entangled. {This is  in line with the results of Ref. \cite{Lo2001} regarding projective operations and entanglement manipulations of entangled states.}

As a corollary, we immediately see the relationship between {a state's entanglement and its symmetries under general local quantum operations}. Fully entangled {bipartite} pure states $\rho=\ket{\psi}\bra{\psi}$ satisfy \eq{\rho\to\sum_l \left(K_l\ot\id\right)\rho \left(K_l\ot\id\right)^\dagger=\sum_l \left(\id\ot J_l\right)\rho \left(\id\ot J_l\right)^\dagger,\label{eq:related Kraus transformation}}
for the related transformations $J_l=\Sigma K_l^T \Sigma^{-1}$. The related map again exists for all quantum operations if and only if the state is fully entangled. Furthermore, when the map on $A$ is completely-positive and trace-preserving (CPTP), then the related map on $B$ is completely positive (see Appendix \ref{app:CPTP}). Conditions for the trace preservation of the related map are more nuanced, requiring both maximal entanglement and the constraint on the Kraus operators $\sum_l K_lK_l^\dag=\id$ in addition to the usual trace preservation condition $\sum_l K_l^\dag K_l=\id$. {This does not hold true in general;} hence, no state guarantees that any CPTP map on subsystem $A$ can be expressed as a CPTP map on subsystem $B$. 
\textcolor{black}{Of course, trace preservation \cite{NielsenChuang2000} and complete positivity \cite{Suarezetal1992,Pechukas1994,HartmannStrunz2019} are not stringent requirements of the quantum theory, so we can say that sensible related transformations as in \eqref{eq:related Kraus transformation} exist for all quantum operations on fully entangled pure states.}


With minor modifications, these results can be adapted to deal with qudits of unequal dimension. We find that the action of an operation on the smaller system can be replicated on the larger system with some freedom if and only if the Schmidt rank of the combined state is equal to the dimension of the smaller system. Actions on the larger system cannot, in general, be replicated by actions on the smaller system.
The related but distinct notions of when nonlocal operations on a bipartite state can have support on just a single subsystem and when the state dynamics are themselves symmetric, have been recently investigated in \cite{Hirai2019} and \cite{Qinetal2020}, respectively.

We compile these ideas into a generalized theorem:
\begin{thm}[Entanglement-Symmetry Relation for Pure States of Two Qudits]\label{thm:1}
\ \\ Let $\ket\psi_{AB} =\sum_{j=1}^{d_{A}}\sum_{k=1}^{d_{B}}C_{jk}\ket{j}_{A} \ket{k}_{B}$
be a pure state of two qudits, $A$ and $B$, of dimension $d_{A}$
and $d_{B}$, respectively, with $d_{A}\leq d_{B}$. Then $\ket{\psi}_{AB} $ is fully entangled if and only if any local operation $U\otimes \id_{d_B}$ acting on $A$ can equivalently be expressed
as a related operation $\id_{d_A}\otimes V$ acting on $B$.

In the case $d_A$ = $d_B$, a state is maximally entangled if and only if every unitary on $A$ can be replicated by a unitary on $B$. 
\end{thm}

\begin{proof}
First, let us use the singular value decomposition $C=Y\Sigma Z$ with $\left | i\right\rangle_{A}=\sum_j Y_{ji} \left | j\right\rangle_{A}$ and $\left | i\right\rangle_{B}=\sum_k Z_{ik} \left | k\right\rangle_{B}$ to write the state in the Schmidt basis,
\begin{equation}\label{eq:SchmidtState3}
\ket\psi_{AB} =\sum_{i=1}^{d_{A}}\sigma_i\left|i\right\rangle_{A} \left|i\right\rangle_{B}.
\end{equation}
By convention $\Sigma$ is a $d_A\times d_B$ diagonal matrix with the Schmidt coefficients $\sigma_i$ listed in decreasing order.

If $\ket{\psi}_{AB}$ is fully entangled, then, by definition, its Schmidt rank is maximal, viz. $\rk(\Sigma)=d_A$. This implies that there exists a right inverse $\Sigma_R^{-1}$ such that $ \Sigma\Sigma_R^{-1}=\id_{d_A}$. We can then construct an operator $\tilde V$ in this basis as
\begin{equation}
\tilde V=(\Sigma_R^{-1} \tilde U \Sigma)^T,
\label{eq:V tilde sigma right}
\end{equation}
where $\tilde U=Y^\dag U Y$ is the original $U$ in the Schmidt basis. By construction $\tilde V$ satisfies the conditions \eqref{eq:matrixconditions2} equivalent to $(\tilde U\ot\id) \ket\psi_{AB}=(\id\ot \tilde V) \ket\psi_{AB}$. When $d_A< d_B$, $\tilde V$ is a highly non-unique solution to this equation. Transforming back to the original basis with $ U=Y \tilde U Y^\dag$, and $ V=Z^T \tilde VZ^*$ we have 
\begin{equation}
V=(C_R^{-1} U C)^T,
\end{equation}
where $C_R^{-1}\equiv Z^\dag \Sigma_R^{-1} Y^\dag$, and it follows that $(U\ot\id) \ket\psi_{AB}=(\id\ot V) \ket\psi_{AB}$.

Conversely, assuming that a solution $V$ exists to the equation $(U\ot\id) \ket\psi_{AB}=(\id\ot V) \ket\psi_{AB}$ for any operation $U$, by changing to the Schmidt basis this also implies $\tilde V$ exists such that $(\tilde U\ot\id) \ket\psi_{AB}=(\id\ot \tilde V) \ket\psi_{AB}$. We have shown above that this is equivalent to $\tilde V$ satisfying
$
U\Sigma=\Sigma{\tilde V}^T$.
Given that $\Sigma$ is diagonal, this can only be true for all $U$ if $\rk(\Sigma)=d_A$, which is the definition of $\ket\psi_{AB}$ being fully entangled.

The second statement in the theorem was proven in the preceding discussion.
\end{proof}

This generalization shows the important role that the dimensions of the subsystems can play for reproducing operations with finite-dimensional systems. One might try to salvage the operational symmetries in the $d_A>d_B$ case by adding an ancillary qudit, $C$, separable from $B$, so that $d_A \leq d_B d_C$, and then applying our Theorem to systems $A$ and $BC$. This will not work since $A$ will no longer be fully entangled with $BC$, even it was initially fully entangled with $B$.

It comes as no surprise that fully entangled states can achieve quantum steering \cite{Zhenetal2016}. The operational symmetries of fully entangled states allow one party's local measurements to be remotely affected by another's local operations,
negating the possibility for the former party to assign a local hidden state model to explain their measurement results. Moreover, there can be no local hidden variable model to explain the measurements obtained by the pair of parties on states exhibiting operational symmetries, in line with the fact that all fully entangled states can be used to violate Bell inequalities \cite{Gisin1991,GisinPeres1992,Yuetal2012}. The hierarchy is as follows: operationally-symmetric states are a (non-convex) subset of states violating Bell inequalities, which are a subset of steerable states, which are a subset of entangled states.


These results can be applied immediately to tripartite pure states $\ket{\psi}_{ABC}$.
Although a generalized Schmidt decomposition for multipartite states does not exist, one does exist for any given bipartition, such as
$
\ket{\psi} _{ABC}=\sum_{i}\sigma_{i}\ket{i}_{AC} \ket{i}_{B}.
$
 Our Theorem then guarantees that an operation $U$ on $A$ can always be represented as a related operation $V$ on $B$ when $d_Ad_C\leq d_B$ and the Schmidt matrix $\Sigma$ has rank $d_Ad_C$. Again, for the related matrix $V$ to be unitary for all unitary $U$, it is necessary and sufficient that all of the Schmidt coefficients are equal. 
 Similarly, for the bipartition $
\ket{\psi} _{ABC}=\sum_{i}\sigma_{i}\ket{i}_{A} \ket{i}_{BC},
$ an operation $U$ on $A$ can be represented as an operation $V$ on $BC$ when $d_A\leq d_Bd_C$ and the Schmidt matrix $\Sigma$ has rank $d_A$.
The unitarity of the related map on $BC$ is  a necessary condition for the existence of a quantum channel on system $B$ alone to reproduce the action of $U$ on system $A$, but it is not sufficient. {Furthermore, the general non-factorizability of $V$ over $B$ and $C$ is a manifestation of the monogamy of entanglement \cite{Coffman2000, Osborne2006}. It is one of the reasons that one cannot use the idea of operational symmetries to reduce the number of local unitaries required to transform between two equivalent graph states, which form an important subset of multipartite entangled states \cite{Hein2004, Nest2004, Tzitrin2018}}.
 The case of multipartite pure states with more than three qudits can be similarly analyzed by selecting an appropriate bipartition for the local operations.

 We can also consider the possibility of finding related operations for mixed states, but unlike pure states the only examples are trivial. Using ideas from the multipartite setting, a bipartite mixed state $\rho_{AB}$ can be studied by considering its purification $\ket\rho_{ABC}$. Since a purification can always be found with $d_{A}d_{B}\leq d_{C}$, we can apply our Theorem by assuming full entanglement in the bipartition $AB\vert C$
or in $A\vert BC$. In the first case, we look for a separable operation $U_A\otimes W_B$ that leaves the state $\rho_{AB}$ invariant, so that $W_B$ reverses $U_A$ (see \cite{Gheorghiu2007} for a discussion of such separable operations). This can occur if and only if $U_A\otimes W_B$ is related to a unitary operation $V_C$, leaving $\rho_{AB}$ unaffected after $C$ is traced out. However, by our Theorem, $V_C$ is unitary if and only if $\ket\rho_{ABC}$ is maximally entangled between $AB$ and $C$, which implies $\rho_{AB}$ is maximally mixed. {This is a trivial case, since any trace-preserving operation on $A$ leaves the state unchanged.}

For the second type of bipartition, $A\vert BC$, we can find a joint operation $V_{BC}$ for every operation on $A$ by choosing a fully entangled purification. However, $V_{BC}$ cannot be expressed as
an operation on $B$ alone in general; the factorization into unitaries $V_{BC}=V_{B}\otimes V_{C}$ is only possible in special cases, for which we find the desired result $U_{A}\rho_{AB} U_{A}^{\dagger}=V_{B}\rho_{AB} V_{B}^{\dagger}$.
Therefore we conclude that the maximally mixed state is the only mixed
state that admits complete operational symmetry.


Seeing that the related operation $V=\Sigma U^T\Sigma^{-1}$ is not generally unitary even for a local unitary $U$ on $A$, one may wonder which unitary on $B$ can best replicate the action of $U$. We demonstrate in Appendix \ref{app:maximizing fidelity} that the optimal unitary comes 
from the unitary part of $\Sigma U^T\Sigma$. For any bipartite pure state, we show that 
\eq{M\left(U \right)\equiv\max_{V\in\mathcal{U}}\left|\bra{\psi}U^\dagger\ot V\ket{\psi}\right|=\Tr\left|\Sigma U\Sigma\right|,
\label{eq:max unitary}
} maximizing the fidelity over the set of unitaries $\mathcal{U}$ on $\mathcal{H}_B$. This function is bounded between 0 and 1, with maximally entangled states achieving the upper bound regardless of $U$. 
Perturbative expansions of the trace norm appearing in \eqref{eq:max unitary} for variations in $U$ or $\Sigma$ can be carried out using recently-developed methods in patterned-matrix calculus \cite{Cresswell2019}.

The operational symmetries enjoyed by fully entangled states motivate new ways to quantify quantum entanglement, a burgeoning area of research \cite{Vidal2002,MintertBuchleitner2007,Islam2015a,Schwaigeretal2015,Lancienetal2016,Cresswell2017a,UmemotoTakayanagi2018,TanJeong2018}.
We look for a metric inspired by \eqref{eq:max unitary} that depends only on the state in question. Maximizing $M\left(U\right)$ over all $U$ yields unity for any state (let $U=\id$), whereas minimizing over $U$ gives the \textit{minimum fidelity}
\eq{m\left(\rho\right)\equiv\min_{U\in \mathcal{U}}\max_{V\in \mathcal{U}}\left|\Tr \left[(U^\dag\ot V)\, \rho\right]\right|,\label{eq:minimum fidelity measure}} which equals $\sum_{i=1}^{d_A}\sigma_i\sigma_{d_A-i}$ for pure states. This leads to a measure that vanishes for states with Schmidt rank $r\leq d_A/2$ and ranges up to 1 for maximally entangled states. Another promising candidate, which we call the $\emph{symmetry of entanglement}$, is sensitive to the entanglement in states with arbitrary Schmidt rank:
\eq{
E_S\left(\rho\right)\equiv\int dU\,\max_{V\in \mathcal{U}}\left|\Tr \left[(U^\dag\ot V)\, \rho\right]\right|,\label{eq:average fidelity measure}
}
where $dU$ is the Haar measure over unitary operators on $\mathcal{H}_A$. 
For pure states this can be simplified using \eqref{eq:max unitary}, and the quantity ranges from $\sqrt{\pi}\Gamma\left(d_A\right)/2\Gamma\left(d_A+1\right)$ for separable states to 1 for maximally entangled states.

We plot our symmetry-inspired entanglement quantifiers for various states in Fig. \ref{fig:comparing4x4_measures}. The functions $m\left(\rho\right)$ and $E_S(\rho)$ are distinct from both the entanglement entropy $S\left(\rho\right)=-\Tr\left(\rho_A\log \rho_A\right)$ and the entanglement negativity  $\mathcal{N}\left(\rho\right)\propto \tr\left\vert\rho^{T_A}\right\vert
-1$, where a superscript $T_i$ denotes the partial transpose with respect to subsystem $i$ \cite{Vidal2002, Cresswell2019}. 
 For pure states, entanglement entropy and negativity reduce to  $-\Tr\left(\Sigma^2\log \Sigma^2\right)$, and  $\left[\tr\left(\Sigma\right)\right]^2 -1$ respectively, and share similar behaviours with our measures as demonstrated in Fig. \ref{fig:comparing4x4_measures}. Going beyond pure states, the minimum fidelity and symmetry of entanglement have desirable properties such as convexity (Appendix \ref{app:maximizing fidelity}), but cannot be considered true entanglement monotones \cite{PlenioVirmani2014} as they are sensitive to classical correlations, i.e., mixedness. 
It may ultimately be possible to use this property of operational symmetries to help distinguish between classical and quantum correlations in a given quantum state \cite{Modietal2012}.
\begin{figure}
    \centering
    \includegraphics[width=\columnwidth]{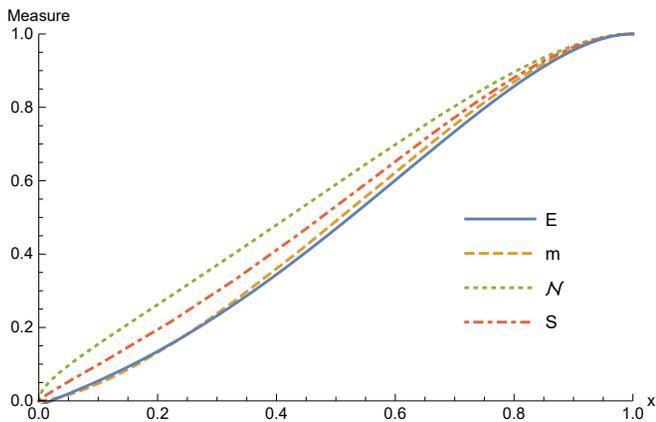}
    \caption{Quantifying entanglement for a state with Schmidt coefficients $\left(1-x/4-x^2/4-x^3/4\right)^{1/2} $, $\tfrac{1}{2}x^{1/2}$, $\tfrac{1}{2}x$, and $\tfrac{1}{2}x^{3/2}$. Plotted are the newly-proposed symmetry of entanglement (blue, solid) and  minimum fidelity (orange, dashed) in comparison to negativity (green, dotted) and entanglement entropy (red, dot-dashed). All measures are scaled such that they range between 0 for separable states and 1 for maximally entangled states.}
    \label{fig:comparing4x4_measures}
\end{figure}


Symmetries under specific unitaries, a restriction of our more general treatment, are exploited in the literature \cite{Zurek2003,Zurek2005} under the moniker of envariance. Envariance is useful for providing a derivation of Born's Rule, since it is thought to be $\emph{ad hoc}$ to postulate this rule in interpretations of quantum theory -- among them Everettian and Many Worlds quantum mechanics \cite{dewitt2015} -- which only subsume unitary evolution and thus have no concept of wavefunction collapse. In one of several approaches in \cite{Zurek2005}, Born's Rule is established with the help of \emph{partial swaps}, which exchange two orthonormal basis sets in some subspace of the system Hilbert space. The ability in maximally entangled states to undo a partial swap on the system by a partial counterswap of the environment implies that a perfect knowledge of the global state comes at the expense of a complete, objective ignorance of the local states. Therefore one is justified in assigning equal probabilities to the local states, and Born's Rule follows from an extension of this idea to states without maximal entanglement.

The consideration of envariance under partial swaps is close in spirit to our work, since it is equivalent to operational symmetry under all unitaries, which we have seen holds for maximally entangled states. In this way, our proposed symmetry of entanglement \eqref{eq:average fidelity measure} can be thought of as quantifying the degree of objective ignorance or indifference an observer possesses about a state.


Some of our results can be considered as quantum mechanical analogues of the Reeh-Schlieder theorem from algebraic QFT, which underlines the highly entangled nature of typical states of quantum fields \cite{Reeh1961}. We briefly review the main ideas of this theorem in order to point out the similarities {to our quantum-mechanical results} (see \cite{Witten2018,Harlow2018} for a more complete introduction and applications).

Consider an algebra $\mathcal{A}_A$ of quantum operators. A state $\ket{\Psi}$ is \emph{cyclic} with respect to $\mathcal{A}_A$ when the states $u\left | \Psi\right\rangle$ for $u\in\mathcal{A}_A$ are dense in the Hilbert space $\mathcal{H}$. For instance, the vacuum state $\ket{\Omega}$ of a QFT is cyclic with respect to the algebra of operators supported on a complete Cauchy slice through the spacetime, since any state in the vacuum sector Hilbert space $\mathcal{H}_0$ can be prepared on an initial value hypersurface.

We say that a state $\ket{\Psi}$ is \emph{separating} with respect to $\mathcal{A}_A$ when $u\ket{\Psi}=0$ implies $u=0$ for any $u\in\mathcal{A}_A$. As an example, consider an  algebra $\mathcal{A}_{B}$ that commutes with $\mathcal{A}_A$. Let $\ket{\Psi}$ be cyclic with respect to $\mathcal{A}_B$, and suppose that there is a $u\in\mathcal{A}_A$ such that $u\ket{\Psi}=0$. It follows that $u v\ket{\Psi}=0$ for all $ v\in\mathcal{A}_B$ by commuting $u$ onto $\ket{ \Psi}$. But, by cyclicity, the states $v\ket{\Psi}$ are dense in $\mathcal{H}$ so that $u$ annihilates all states and must be the zero operator $u=0$. The implication is that a cyclic state for $\mathcal{A}_{B}$ will be separating for any other algebra of commuting operators $\mathcal{A}_{A}$. 

In this language, the Reeh-Schlieder theorem says that the vacuum state of a QFT on Minkowski spacetime is cyclic and separating with respect to any local algebra $\mathcal{A}_\mathcal{V}$ of field operators supported in an open neighbourhood $\mathcal{V}$ of a spacetime region. This surprising theorem has a very pertinent corollary for spacelike separated neighbourhoods $\mathcal{V}$ and $\mathcal{V}'$. If $u\in \mathcal{A}_\mathcal{V}$ is any operator supported in $\mathcal{V}$ then there exists an operator $v\in\mathcal{A}_{\mathcal{V}'}$ supported in $\mathcal{V}'$ such that $u \ket{ \Omega}= v\ket{ \Omega}$, since either local algebra generates $\mathcal{H}_0$ by cyclicity. We are free to choose $\mathcal{V}'$ arbitrarily small and to be located at the other end of the universe from $\mathcal{V}$, yet the approximation between the states $u \ket{ \Omega}$ and $v\ket{ \Omega}$ can be made arbitrarily good. It is important to note that the related operator $v$ will, in general, not be unitary.

There is a clear parallel between the cyclicity of the vacuum in QFT and the notion of fully entangled states we have used. Returning to the exemplary case of two identical qudits \eqref{eq:compState}, we can consider the algebra $\mathcal{A}_A$ of linear operators acting on $\mathcal{H}_A$, and the commuting algebra $\mathcal{A}_B$ of operators acting on $\mathcal{H}_B$. Assuming the state $\left|\psi\right\rangle=\sum_{i=1}^{d}\sigma_{i}\left|i\right\rangle_{A}\left|i\right\rangle_{B}$ is fully entangled, $\sigma_i\neq 0$, the only way for $U\ot \id$ to annihilate the state is if $U\ket{i}_A=0$  $\forall\ i$. But since the $\ket{i}_A$ form a complete, orthonormal basis of $\mathcal{H}_A$, this can only occur if $U$ is the zero operator. Were the state not fully entangled, we could find a non-zero operator with kernel corresponding to the support of $\ket{\psi}$. Hence, the fully entangled condition for $\ket{\psi}$ implies that it is separating for $\mathcal{A}_A$, and a separating vector for $\mathcal{A}_A$ is cyclic for the commutant of $\mathcal{A}_A$ \cite{Witten2018}, which is $\mathcal{A}_B$. In the equal-dimensional case it is clear that $\ket{\psi}$ is also separating for $\mathcal{A}_B$, and therefore cyclic for $\mathcal{A}_A$. From this point of view it should not be surprising that the related operators $V^B$ in \eqref{eq:UAaction} exist for any $U^A$ if and only if $\ket{\psi}$ is fully entangled.

However, we have shown that it is still possible to find related operations even when the subsystems are not of equal size, as long as the state has maximal Schmidt rank. The state \eqref{eq:SchmidtState3} is separating for $\mathcal{A}_A$, but not for $\mathcal{A}_B$. Still, this implies that it is cyclic for $\mathcal{A}_B$, which ensures that any local operation on $A$ can be reproduced as a related operation on $B$. We have provided an explicit construction of such related operators in our Theorem.

{One other notable application of our results is to the so-called \emph{classical entanglement} between spatial and polarization degrees of freedom of an electromagnetic wave \cite{Spreeuw1998,Toppeletal2014,KarimiBoyd2015,Eberlyetal2016}.} 
The polarization subspace is two-dimensional and the spatial subspace is an infinitely-continuous stochastic function space \cite{Qianetal2015}; our Theorem thus ensures that any operator acting on the polarization degree of freedom of a classically-entangled state $\mathbf{E}\left(\mathbf{r}; t\right)=E_x\mathbf{x}+E_y\mathbf{y}$ can be replicated by an operator on the spatial degree of freedom. This classical example is intuitive, as an operation on the polarization degree of freedom corresponds to a physical rotation, which can be equivalently realized as a rotation of the functions in the superposition. Our results are a consequence of tensor product structure and group theory that extend beyond quantum mechanics.
\color{black}

While entanglement is \textit{defined} in terms of properties of the state, it is often \textit{explained} in terms of operations acting on the state. In this paper, we have characterized entangled states in terms of their symmetries under local operations. This allowed us to establish new measures of entanglement and gives an expanded framework for understanding concepts such as steering, envariance, and the Reeh-Schlieder theorem. {Operational symmetries and the Reeh-Schlieder theorem in particular have profound implications that warrant future investigation by the broader quantum information community.}

\begin{acknowledgments}
	We would like to thank Daniel James for inspiring this work, and Eli Bourassa and Hoi-Kwong Lo for helpful discussions. I.T. is supported by an Ontario Graduate Scholarship and a Discovery Grant from the Natural Sciences and Engineering Research Council of Canada (NSERC). A.Z.G. is supported by an Alexander Graham Bell Scholarship (No. 504825) and an NSERC Discovery Grant. J.C.C. is supported by a Vanier Canada Graduate Scholarship and an NSERC Discovery Grant. 
\end{acknowledgments}
\normalfont\textsuperscript{1}All authors contributed equally to this work.

\appendix
\section{Symmetries of completely positive trace-preserving maps}
\label{app:CPTP}
For a fully entangled pure state $\rho$ on $\mathcal{H}_A\otimes\mathcal{H}_B$ with equal dimensions, \textit{any} operator $K$ acting on the $A$ subsystem can be expressed by a related operator $ \Sigma K^T \Sigma^{-1}$ on $B$, where $\Sigma$ is the matrix with Schmidt coefficients along the diagonal. We thus have, for $J_l=\Sigma K_l^{T}\Sigma^{-1}$,
\eq{\sum_l \left(K_l\ot\id\right)\rho \left(K_l\ot\id\right)^\dagger=\sum_l \left(\id\ot J_l\right)\rho \left(\id\ot J_l\right)^\dagger.} We now determine when the related map is completely positive (CP) and when it is trace preserving (TP).

Unitary operations are TP. When the map on subsystem $A$ is unitary, the only states guaranteeing all related maps to be unitary are the maximally entangled ones; to guarantee that all TP maps on subsystem $A$ can be represented by TP maps on $B$, a state must have $\Sigma\propto \id$. The assumption that the generalized initial map is trace-preserving can be written as $\sum_l K_l^{\dagger} K_l=\id$. Inspecting the same condition for the related map, and using $\Sigma\propto \id$ yields
\eq{\sum_l J_l^\dagger J_l=\id \iff 
\sum_l K_l^{*}K_l^{T}=\id \, .
} This is the condition that the set $\{K_l^{T}\}$ is a valid set of Kraus operators, which is a viable possibility but does not occur in general. It is thus impossible to find a state that guarantees all TP maps on subsystem $A$ to be representable by TP maps on subsystem $B$.

For a quantum operation to be CP its Choi matrix must be positive and thus Hermitian.
We construct the Choi matrices of the maps on $A$ and $B$ by \eq{C_A&=\sum_{i,j}\ket{i}\bra{j}_C\ot
\sum_l K_l \ket{i}\bra{j}_A K_l^{\dagger},\\
C_B&=\sum_{i,j}\ket{i}\bra{j}_C\ot
\sum_l J_l \ket{i}\bra{j}_B J_l^{\dagger},
} using an auxiliary space $\mathcal{H}_C$ of the same dimension. These matrices are both Hermitian, by the identity $\left(A\ot B\right)^\dagger=A^\dag\ot B^\dag$ and a relabelling of indices $i$ and $j$. 

To show that a Hermitian matrix is positive it suffices to show that its expectation values are positive within some complete basis; as usual we use the Schmidt basis for convenience. The initial map being CP implies that $\sum_l \bra{m}K_l\ket{n}\bra{n} K_l^{\dag}\ket{m}\geq 0$ for all $m,n$. In the related case, we see
\eq{
&\sum_l \bra{m}\Sigma K_l^{T}\Sigma^{-1}\ket{n}\bra{n}\Sigma^{-1} K_l^{*}\Sigma\ket{m} \\
=&\sum_l \bra{m}\sigma_m K_l^{T}\sigma_n^{-1}\ket{n}\bra{n}\sigma_{n}^{-1} K_l^{*}\sigma_m\ket{m}\geq 0\\
\iff &\sum_l \bra{m} K_l^{T}\ket{n}\bra{n} K_l^{*}\ket{m}\geq 0\\
\iff &\sum_l \bra{n} K_l^{\dag}\ket{m}\bra{m} K_l^{}\ket{n}\geq 0,
} which is the assumed condition. Thus the related map on a fully entangled state is CP given that the initial map is CP. This result can be extended to the case where $d_A <d_B$ using the right inverse discussed in \eqref{eq:V tilde sigma right} of the main text.

\section{Maximizing fidelity with unitary operations}
\label{app:maximizing fidelity}
Given a fully entangled bipartite state $\left|\psi\right\rangle =\sum_{i=1}^{d}\sigma_i\left|i\right\rangle \left|i\right\rangle $
where $d=d_{A}=d_{B}$ for clarity of the derivation, we know that a unitary $U$ on subsystem $A$ will
have a related operation $V=\Sigma U^T \Sigma^{-1}$ on $B$. The related operation takes $\ket{\psi}$ to a normalized quantum state; however, when acting on other states $V$ will not, in general, be trace-preserving. This prompts the question of which unitary operation $V$ maximizes the fidelity between $U\otimes\id\ket{\psi}$ and $\id\ot V\ket{\psi}$ for a given $\ket{\psi}$ and $U$.

\subsection{Fully-entangled states}
We first start with states whose Schmidt matrices $\Sigma$ are full rank. Given an operation $U$ on $A$, we would like to find
\eq{M\left(U\right)=\max_V\left|\bra{\psi} U^\dag\ot V\ket{\psi}\right|,}
where the maximum is taken over all unitary operations. Full entanglement implies
\eq{M\left(U\right)&=\max_V\left|\bra{\psi} \id\ot V \Sigma U^*\Sigma^{-1}\ket{\psi}\right|\\
&=\max_V\left|\sum_{i}^d\bra{i}  V \Sigma U^*\left(\sum_{k=1}^d\sigma_k^{-1}\ket{k}\bra{k}\right)\sigma_i^2\ket{i}\right|\\
&=\max_V\left|\Tr\left(V\Sigma U^*\Sigma\right)\right|.
\label{eq:maximize unitary V derivation}
}
The von Neumann operator trace inequality tells us that $\left|\Tr\left(X\right)\right|\leq\Tr\left|X\right|$, and $\left|V\Sigma U^*\Sigma\right|=\left|\Sigma U^*\Sigma\right|$ is independent of $V$, so we can choose $V$ to be the inverse of the unitary part of $\Sigma U^*\Sigma$ to yield
\eq{M\left(U\right)=\Tr\left|\Sigma U^*\Sigma\right|=\Tr\left|\Sigma U\Sigma\right|,} for any fully entangled state. This is easy to compute for maximally entangled states:
\eq{M\left(U\right)&=\frac{1}{d}\Tr\left|U\right|
=1,\quad \forall\ U.}
For other states, it is possible to achieve $M=1$ {only} for certain values of $U$. This prompts a measure of entanglement for a given state that minimizes $M$ over all unitaries $U$. The minimum is achieved when $U$ is an anti-diagonal matrix with ones on the diagonal, causing the largest and smallest Schmidt coefficients to be paired together:
\eq{m=\min_{U\in \mathcal{U}}\Tr\left|\Sigma U\Sigma\right|=\sum_{i=1}^d \sigma_i\sigma_{d-i}.}
However, this measure will not be viable for states with Schmidt rank less than $d/2$, so we will explore a new measure after discussing how to calculate $M$ for such states.

\subsection{Arbitrary pure states}
We would now like to calculate
\eq{M\left(U\right)=\max_V\left|\bra{\psi} U^\dag\ot V\ket{\psi}\right|,}
for unitaries $U$ and states $\ket{\psi}$ whose Schmidt matrices $\Sigma$ are not necessarily full rank. Arranging the bases such that the nonzero Schmidt coefficients are in the first $r\times r$ blocks of their corresponding matrices, we have that
\eq{M\left(U\right)=\max_V\left|\bra{\psi} \begin{pmatrix}u_r^\dag&\mathbf{0}\\\mathbf{0}&\id_{d-r}\end{pmatrix}\ot \begin{pmatrix}v_r&\mathbf{0}\\\mathbf{0}&\id_{d-r}\end{pmatrix}\ket{\psi}\right|.} The matrices $u_r$ and $v_r$ are the first $r\times r$ blocks of $U$ and $V$, respectively, and are not necessarily unitary, but one can always find $V$ such that $v_r$ is unitary. 

The entire derivation of \eqref{eq:maximize unitary V derivation} now holds within the Schmidt subspace, where we now use $\Sigma_r$ to represent the first $r\times r$ block of $\Sigma$:
\eq{M\left(U\right)&=\max_V\left|\bra{\psi} u_r^\dag\ot v_r\ket{\psi}\right|\\
&=\max_V\left|\bra{\psi} \id\ot v_r \Sigma_r u_r^*\Sigma_r^{-1}\ket{\psi}\right|\\
&=\max_V\left|\Tr\left(v_r\Sigma_r u_r^*\Sigma_r\right)\right|.
}
Again, the absolute value of the trace is maximized by taking $v_r$ to be a unitary matrix inverting the unitary part of $\Sigma_r u_r^* \Sigma_r$, yielding the same result as for fully entangled states:
\eq{M\left(U\right)=\Tr\left|\Sigma_r u_r^T\Sigma_r\right|=\Tr\left|\Sigma U\Sigma\right|.} 

\subsection{Mixed states}
The above generalizes to bipartite mixed states by purifying the system with an {ancillary system} $C$ and using the Schmidt decomposition for the bipartition $A|BC$. However, now the maximization over $V$ does not guarantee that the correct unitary matrix can be found. Even for a maximally entangled purified state, the resulting condition is
\eq{M\left(U\right)&=\max_V\left|\bra{\psi_{ABC}}U^\dag\ot V\ot \id_C\ket{\psi_{ABC}}\right|\\
&=\max_V\left|\Tr\left[\left(V_B\ot \id_C\right)\left(\Sigma U^*\Sigma\right)_{BC}\right]\right|\\
&\neq \Tr\left|\Sigma U^*\Sigma\right|,
} with the inequality stemming from the lack of freedom within $V_B\ot\id_C$ to achieve arbitrary unitaries on $BC$.

If we instead maximized over unitary operators on $BC$, we would simply recover the result that $M\left(U\right)=\Tr\left|\Sigma U\Sigma\right|$, where $\Sigma$ is the matrix of Schmidt coefficients of the $A|BC$ partition of the purified state.

\subsection{Averaging over $U$} A suitable way to measure entanglement is by averaging $M\left(U\right)$ over all $U$ using the normalized Haar measure $dU$.
Separable states have $M\left(U\right)=\left|U_{11}\right|=\left|\bra{\psi}U\ket{\psi}\right|$, requiring the average of a single element of a unitary over the Haar measure $dU$.
This can be done using the fact that each element $U_{jk}=r_{jk}e^{i\theta_{jk}}$ of a unitary matrix has the same distribution, namely \cite{PetzReffy2004} \eq{dU_{jk}=\frac{d-1}{\pi}\left(1-r_{jk}^2\right)^{d-2} r_{jk}\, dr_{jk}\, d\theta_{jk},} where $\theta_{jk}\in\left[0,2\pi\right]$ and  $r_{jk}
\in \left[0,1\right]$. For $\left|U_{11}\right|$, we get \eq{\int dU \left|U_{11}\right|=\frac{\sqrt{\pi } \Gamma (d)}{2 \Gamma \left(d+\frac{1}{2}\right)}
= \frac{2^{2d-2}}{2d-1}\binom{2d-2}{d-1}^{-1}.}

If we try to deviate slightly from a separable state, with $\Sigma=\text{diag}\left(\sqrt{1-\epsilon},\sqrt{\epsilon},0,\cdots\right)$, we have 
\eq{\label{eq:M 2x2}
M\left(U\right)&=\Tr\left|\begin{pmatrix}
U_{11}\left(1-\epsilon\right) & U_{12}\sqrt{\epsilon}\sqrt{1-\epsilon}\\
 U_{21}\sqrt{\epsilon}\sqrt{1-\epsilon} & \epsilon U_{22}
\end{pmatrix}\right|.
} To lowest order in $\epsilon$ the eigenvalues of this matrix are $U_{11}+\tfrac{U_{12}U_{21}}{U_{11}}\sqrt{\epsilon}$ and $\left(U_{22}-\tfrac{U_{12}U_{21}}{U_{11}}\right)\sqrt{\epsilon}$. This yields \eq{M\left(U\right)&\approx \left|U_{11}\right|\\
+\sqrt{\epsilon}&\left(\left|U_{22}-\frac{U_{12}U_{21}}{U_{11}}\right|
+\frac{1}{\left|U_{11}\right|}\Re \left[\frac{U_{11}^*U_{12}U_{21}}{U_{11}}\right]\right).} 
The $\mathcal{O}\left(\sqrt{\epsilon}\right)$ term always averages to a positive number because the first term in parentheses is always positive, and the second term in parentheses is proportional to $\cos\left(\theta_{12}+\theta_{21}-2\theta_{11}\right)$, which averages to zero. We see that separable pure states are a minimum of our measure. 

We relax the restrictions of small $\epsilon$ and plot our entanglement quantifier versus $\epsilon$ for various dimensions in Fig. \ref{fig:comparing248}. The integrand \eqref{eq:M 2x2} is the same for each dimension $d$ and the integrals over extraneous components of $U$ give unity, but the Haar measure changes with $d$. This makes the measure
\eq{
E=\int dU M\left(U\right),
}
decrease monotonically with $d$, as does the normalized measure
\eq{\label{eq:normalized measure}
\mathcal{E}=\frac{2 \Gamma \left(d+\frac{1}{2}\right)E-\sqrt{\pi } \Gamma (d)}{2 \Gamma \left(d+\frac{1}{2}\right)-\sqrt{\pi } \Gamma (d)}.} The measure $m=\sum_{i=1}^d\sigma_i\sigma_{d-i}$ also decreases monotonically with $d$, but cannot distinguish between states with Schmidt rank less than $d/2$.
\begin{figure}
\centering
    \includegraphics{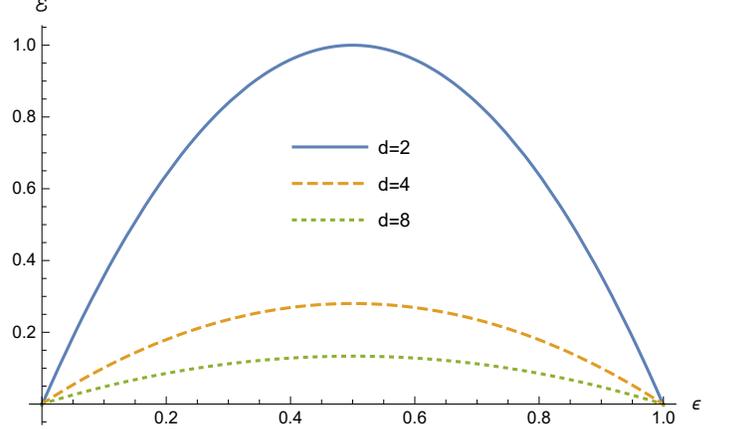}
    \caption{Entanglement in the $d$-dimensional state with Schmidt coefficients $\left(\sqrt{1-\epsilon},\sqrt{\epsilon},0,\cdots\right)$ versus $\epsilon$. For $d=2$ (blue, solid) the normalized measure \eqref{eq:normalized measure} ranges from 0 for separable states ($\epsilon=0,1$) to 1 for maximally entangled states ($\epsilon=1/2$). For $d=4$ (orange, dashed) and $d=8$ (green, dotted) the state can never be maximally entangled. The ordering of the entanglement contained between states with different $\epsilon$ is independent of $d$, and the measure decreases monotonically with $d$.}
    \label{fig:comparing248}
\end{figure}

\subsection{Convexity} The convexity of our entanglement quantifiers $m\left(\sum_i p_i\rho_i\right)\leq\sum_ip_im\left(\rho_i\right)$ and $E_S\left(\sum_i p_i\rho_i\right)\leq\sum_ip_iE_S\left(\rho_i\right)$ follows from the triangle inequality
\eq{
\max_V
\left|\sum_ip_i\Tr\left(U\ot V\,\rho_i\right)\right|\leq \sum_ip_i\max_V\left|\Tr\left(U\ot V\,\rho_i\right)\right|\\
\leq \sum_ip_i\max_{V_i}\left|\Tr\left(U\ot V_i\,\rho_i\right)\right|.
}

\end{document}